\documentclass[runningheads]{llncs}
\usepackage[T1]{fontenc}
\usepackage{environ}

\usepackage[ruled]{algorithm2e}
\usepackage{xcolor}
\usepackage{comment}
\usepackage{booktabs}
\usepackage{multicol}
\usepackage{graphicx,amsmath,amssymb}
\usepackage{cleveref}
\usepackage{fullpage}
\newcommand{\repeattheorem}[1]{%
  \begingroup
  \renewcommand{\thetheorem}{\ref{#1}}%
  \renewcommand{\label}[1]{}%
 \expandafter\expandafter\expandafter\theorem
  \csname reptheorem@#1\endcsname
  \endtheorem
  \endgroup
}

\NewEnviron{reptheorem}[1]{%
  \global\expandafter\xdef\csname reptheorem@#1\endcsname{%
    \unexpanded\expandafter{\BODY}%
  }%
  \expandafter\theorem\BODY\unskip\label{#1}\endtheorem
}

\title{Dynamic Breadth First Search with Predictions}

\titlerunning{Dynamic BFS with Predictions}

\author{Shahbaz Khan \and Shubham Kumar Verma \and Utkarsh Lohiya}
\authorrunning{S. Khan, S. K. Verma, and U. Lohiya}
\institute{Department of Computer Science and Engineering,\\
Indian Institute of Technology Roorkee, India\\
\email{shahbaz.khan@cs.iitr.ac.in}\\
\email{shubham\_kv@cs.iitr.ac.in}\\
\email{utkarsh\_l@cs.iitr.ac.in}}

\newcommand{\UP}[1]{\ensuremath{\text{\texttt{UP}}[#1]}}

\newcommand{\LLq}[1]{\ensuremath{\mathcal{L}_{\mathrm{L}}[#1]}}
\newcommand{\LRq}[1]{\ensuremath{\mathcal{L}_{\mathrm{R}}[#1]}}
\newcommand{\Par}[2][]{\ensuremath{\texttt{Par}_{#1}[#2]}}
\newcommand{\Level}[2][]{\ensuremath{\texttt{L}_{#1}[#2]}}
\newcommand{\dg}[2][]{\ensuremath{\texttt{deg}_{#1}(#2)}}
\newcommand{\In}[1]{\ensuremath{\texttt{In}[#1]}}
\newcommand{\Out}[1]{\ensuremath{\texttt{Out}[#1]}}

\newcommand{\hLevel}[2][]{\ensuremath{\hat{\texttt{L}}_{#1}[#2]}}
\newcommand{\hPar}[2][]{\ensuremath{\hat{\texttt{Par}}_{#1}[#2]}}

\begin{document}

\maketitle

\begin{abstract}
Given a directed graph $G=(V,E)$ having $n$ vertices and $m$ edges, we investigate the problem of maintaining its Breadth-First Search (BFS) tree from a fixed source $s\in V$, subject to an online sequence of $u$ edge insertions and/or deletions in the prediction model.
Our approach leverages a \emph{predicted} update sequence 
that can be preprocessed to build data structures based on anticipated updates, aiding online processing. We present algorithms for incremental (insertions-only), decremental (deletions-only), and fully dynamic (insertions and deletions) settings that maintain a BFS tree (parent and level information). 
Classically, without using predictions, the incremental and decremental BFS tree can be maintained in total $O(mn)$ time using ES Trees~[JACM81], resulting in amortized $O(n)$ update time. However, in the worst case, it requires $\Omega(m)$ time. Moreover, for maintaining BFS trees under fully dynamic edge updates, no algorithm better than trivial recomputation after every update is known, requiring $O(m)$ time per update. Further, assuming the combinatorial BMM conjecture, no algorithm can maintain incremental or decremental BFS polynomially faster than $O(mn)$~[FOCS14], even when the entire update sequence is known in advance~[STOC15].

Our complexity bounds are expressed in terms of error measures of the prediction,
where \textit{error vertices} are those having incorrectly predicted distance, with the corresponding difference called their \textit{error}. We define the vertex prediction error $\eta_{v}$ as the sum of the degree of \textit{error vertices}, and the weighted vertex prediction error $\eta^*_{v}$ as the \textit{error weighted} sum of the degree of \textit{error vertices}. Additionally, we define $\eta_e$ as the number of incorrectly predicted edge updates.

For incremental and decremental BFS our algorithm requires respectively $O(\eta_v + \eta_e)$ and $O(\min\{m,\eta^*_v  + \eta_e\})$ worst case update times using $O(mn)$ preprocessing time and space, and total update time of $O(\eta^*_v  + \eta_e)$.
We also present the space-efficient extensions of the above algorithms requiring $O(n^2)$ and $O(\frac{mn}{k}+n^2)$ space for any $k\in[1,n]$, at the expense of additional $O(n\log n)$ and $O(nk)$ time in the worst case and total time, respectively.
For fully-dynamic updates, our algorithm requires an update time of $O(\min\{m,\eta^*_v+\eta_e\})$ in the worst case. At its core, we extend the classical ES Trees~[JACM81] for batch updates and fully dynamic updates, which may be of independent interest. Notably, this simple extension of ES Trees for batch updates is sufficient to give a competitive prediction algorithm, which can thus be generalized to other graph problems. Finally, we consider error correction mechanisms to reduce dependence of error measures on the first difference between predicted and real updates, using Set Matching, Prefix Set, and Set Edit Distance based error measures. \\

\end{abstract}

\newpage

\section{Introduction}
\label{sec:introduction}
The area of dynamic graph algorithms deals with graphs that evolve by a sequence of online updates, i.e., insertion or deletion of edges. The goal is to maintain the data structure or solution after every such update, much faster than trivially recomputing it from scratch after every update. An algorithm is called \textit{fully dynamic} if it allows both insertion and deletion updates. It is called \textit{partially dynamic}, i.e., \textit{incremental} or \textit{decremental} when the updates are limited to insertions or deletions, respectively. 
Several graph problems are studied in this model~\cite{HolmLT01,even-shiloach,DemetrescuI04}, where many are proven optimal under popular conjectures~\cite{RodittyZ11,AbboudW14,HenzingerKNS15}.


In this paper, we study the problem of maintaining \textit{Breadth First Search (BFS) Trees} (or shortest path trees for unweighted graphs) from a given source under dynamic updates. The applications of BFS trees range from network routing and distance estimation to social-influence analysis and web crawling, which fundamentally maintain BFS distances under a stream of edge updates. Given a graph $G(V,E)$ having $n$ nodes and $m$ edges, its BFS tree from any source $s\in V$ can be computed in $O(m+n)$ time. Recomputing the BFS tree from scratch after each update is thus prohibitive when updates arrive at high frequency on large
networks. The classical result by Even and Shiloach~\cite{even-shiloach} maintains a partially dynamic BFS tree in total $O(mn)$ or amortized $O(n)$ update time, a bound that remained the state of the art for directed graphs for decades. Later, Roditty and Zwick~\cite{RodittyZ11} showed that no combinatorial algorithm can improve the bound using the All Pairs Shortest Paths Conjecture. Further, it was shown that no polynomial improvement for partially dynamic BFS was possible using the combinatorial BMM conjecture by Abboud and Williams~\cite{AbboudW14}, even for offline updates (update sequence known in advance)~\cite{HenzingerKNS15}. 
However, in the worst case, the update time can still be $O(m)$. This was also shown to be tight by Henzinger et al.~\cite{HenzingerKNS15} using the Online Matrix Vector product (OMv) conjecture. 
For the fully dynamic setting, no combinatorial algorithm is known to improve the trivial recomputation in $O(m)$ time, which has also shown to be tight~\cite{AbboudW14}. Allowing algebraic techniques, Brand and  Karczmarz~\cite{BrandK23} improve the worst case update time to $O(n^{1.949})$ for very dense graphs.
In practice, a fully dynamic version of ES trees performs well despite having no better theoretical bounds~\cite{hanauer2020fully}.

We study the classical dynamic BFS tree problem in the prediction model.
A motivating observation about real-world update sequences is that they are
rarely adversarial: network changes tend to follow patterns, user behaviour is often repetitive, and future updates can frequently be anticipated from
historical data or learned models. The emerging paradigm of \emph{algorithms
with predictions} (also called \emph{learning-augmented algorithms}) formalises
this intuition by equipping classical algorithms with a \emph{prediction}—an
auxiliary input, potentially produced by a machine-learning oracle. The goal is to determine if predictions enable provably faster algorithms while retaining correctness
guarantees regardless of prediction quality. The key desired properties of such algorithms are \emph{consistency} (optimal performance when predictions are perfect),
\emph{robustness} (no worse than the best prediction-free algorithms even for the worst prediction), and
\emph{smooth degradation} (performance deteriorates gracefully between the two
extremes). In recent years, the area of algorithms with predictions has received significant attention, particularly for online problems and data structures (see~\cite{MV21,MV22} for a survey), as well as for dynamic graph algorithms. Brand et al.~\cite{brand2024dynamic} presented algorithms for dynamic transitive closure and approximate all pairs shortest paths, while McCauley et al.~\cite{mccauley2024incremental} presented algorithms for incremental topological ordering and cycle detection. Liu and Srinivas~\cite{liu2024predicted} presented a general framework applicable to a plethora of problems. Henzinger et al.~\cite{henzinger2024complexity} mostly focused on lower bounds conditioned on the OMv conjecture. Notably, these results employed different \textit{error measures} to analyze their algorithms.


\subsubsection*{Related Work.}
An interesting variant of dynamic graph model is the offline dynamic model~\cite{EPPSTEIN94,KL15,PengSS19}, where updates are known in advance. This resembles algorithms with predictions for perfect predictions. In fact, most conditional lower bounds for dynamic graph problems also hold for this simpler variant~\cite{brand2024dynamic}. 

A related problem for weighted graphs is the single source shortest path (SSSP), and its all source version: all pairs shortest path (APSP). For positively weighted graphs, 
King~\cite{King99} extended the ES trees for partially dynamic SSSP to require $O(mnW)$ total time, where $W$ is the maximum weight of an edge. However, fully dynamic SSSP cannot be solved faster than simply recomputing the solution from scratch, which also has a matching lower bound~\cite{HenzingerKNS15}. Incremental APSP can be solved in total $O(n^3W)$ update time~\cite{AusielloIMN91}, where $W$ is the maximum weight of an edge. Decremental and fully dynamic APSP, respectively, require total update time of $O(n^3S)$ and $O(mn^2\log^4n)$~\cite{DemetrescuI04}, where $S$ is the number of distinct edge weights. 

Another related problem is the single source reachability (SSR), and all pair reachability (or transitive closure (TC)). Incremental SSR and TC can be respectively maintained in  $O(m)$ and $O(mn)$ total time~\cite{italiano86}. Whereas, the decremental SSR is no better than decremental BFS, requiring $O(mn)$ total time~\cite{even-shiloach}. 
Also, decremental TC can be maintained in total $O(mn)$ time~\cite{Lacki13,Roditty13}. For fully dynamic updates, none of these problems improves over trivially recomputing from scratch, which also has a matching lower bound~\cite{AbboudW14}. Further, all these algorithms also require $\Omega(m)$ worst-case update time, having matching lower bounds~\cite{HenzingerKNS15}.


\subsubsection*{Our results.}
To analyze our proposed algorithms, we describe error measures for the predictions. The vertices with incorrectly predicted shortest distance from $s$ are called \textit{error vertices}, with the corresponding difference called their \textit{error}. The vertex prediction error $\eta_{v}$ is the sum of degrees of \textit{error vertices} (hence $\eta_v\leq m$), and the weighted vertex prediction error $\eta^*_{v}$ is the \textit{error weighted} sum of degrees of \textit{error vertices} (hence $\eta^*_v\leq mn$). The edge prediction error $\eta_e$ is the number of incorrectly predicted updates (hence $\eta_e\leq m$). 




\begin{enumerate}
    \item 
\textbf{Dynamic BFS algorithms in the prediction model.}
We first describe our results for maintaining incremental, decremental, and fully dynamic BFS in the prediction mode.


\begin{reptheorem}{incBFS}[Incremental BFS]
    Given a graph $G(V,E)$ having $n$ vertices and $m$ edges undergoing edge insertions with a predicted sequence of updates, the BFS tree can be maintained in worst case $O(\eta_e+\eta_v)$ time using $O(mn)$ preprocessing time and space, with total update time $O(\eta^*_v+\eta_e)$.
\label{thm:incBFS}
\end{reptheorem}

\begin{reptheorem}{decBFS}[Decremental BFS]
    Given a graph $G(V,E)$ having $n$ vertices and $m$ edges undergoing edge deletions with a predicted sequence of updates, the BFS tree can be maintained in worst case $O(\min\{m,\eta_e+\eta^*_v\})$ time using $O(mn)$ preprocessing time and space and total time $O(\eta^*_v+\eta_e)$.
\label{thm:decBFS}
\end{reptheorem}

\begin{remark}
Due to a polynomially tight lower bound $\Omega(mn)$ for partially dynamic BFS even for offline updates~\cite{HenzingerKNS15}, any algorithm requires $\Omega(mn)$ preprocessing time considering the case of optimal prediction. 
\end{remark}

\begin{remark}
These are smooth degradation of the ES Trees~\cite{even-shiloach} requiring $O(m)$ time in the worst case, and total time 
 $O(\eta_e+\eta^*_v)=O(mn)$.
\end{remark}


   

\begin{reptheorem}{fdBFS}[Fully Dynamic BFS]
    Given a graph $G(V,E)$ having $n$ vertices undergoing $u$ edge updates with a predicted sequence of updates having maximum $m$ edges overall, the BFS tree can be maintained in worst case $O(\min\{m,\eta_e+\eta^*_v\})$ time using $O(mu)$ preprocessing time and space.
\label{thm:fdBFS}
\end{reptheorem}

\begin{remark}
Due to polynomially tight lower bound~\cite{BrandK23,AbboudW14}, any combinatorial algorithm requires equivalent preprocessing time in the worst case, considering the case of optimal predictions. 
\end{remark}

\begin{remark}
It is a smooth degradation of the static BFS algorithm requiring $O(m)$ time per update. 
\end{remark}

\item
\textbf{Space optimizations for partially dynamic algorithms.}
Since our algorithms require space significantly larger than $O(m)$ of ES Trees~\cite{even-shiloach}, we also optimize the space utilization of our partially dynamic algorithm as follows.

\begin{reptheorem}{spaceOpt}[Space Optimization for Incremental BFS]
 Given a graph $G(V,E)$ having $n$ vertices and $m$ edges undergoing edge insertions with a predicted sequence of updates, the BFS tree can be maintained in worst case $O(\eta_e+\eta_v+n\log n)$ time using $O(mn)$ preprocessing time and $O(n^2)$ space, with total update time $O(\eta^*_v+\eta_e+n\log n)$. 
\label{thm:spaceOpt}
\end{reptheorem}

\begin{reptheorem}{spaceOptDec}[Space Optimization for Decremental BFS]
 Given a graph $G(V,E)$ having $n$ vertices and $m$ edges undergoing edge deletions with a predicted sequence of updates, the BFS tree can be maintained in worst case $O(\eta_e+\eta_v^*+nk)$ time using $O(mn)$ preprocessing time and $O\!\left(\frac{mn}{k}+n^2\right)$ space, for any $1\leq k\leq n$, with total update time  $O(\eta^*_v+\eta_e +nk)$. 
\label{thm:spaceOptDec}
\end{reptheorem}

\begin{remark}
Thus, all our algorithms satisfy the desired properties of prediction algorithms, namely, consistency, robustness, and smooth degradation. 
\end{remark}

\item
\textbf{Error correction mechanisms.}
Since our algorithms are strongly dependent on the error measures $\eta_e,\eta_v$ and $\eta^*_v$, we present 
error correction measures to reduce the dependence on the first difference between real and predicted updates. We thus relax them using \textit{Set Matching, Prefix Set} and \textit{Set Edit Distance} based error measures. 

\begin{reptheorem}{trivErr}[Trivial Error Correction]
 Our algorithms perform trivial error correction if predicted updates are permutation of real updates, using $O(1)$ overhead cost per update.
\label{thm:trivErr}
\end{reptheorem}

\begin{remark}
     This resets the error to zero when the update sequence is a simple permutation of the predicted sequence, since the order of updates does not matter henceforth.
\end{remark}

\begin{reptheorem}{nonTrivPDErr}[Non-Trivial Partially Dynamic Error Correction]
Our incremental and decremental algorithms can perform non-trivial error correction, reducing $\eta_e$ by the longest prefix of predicted updates present in the real updates, using $O(1)$ overhead cost per update.
\label{thm:nonTrivPDErr}
\end{reptheorem}


\begin{remark}
     This reduces the error to the first predicted update not present in the real update, due to limitations of the partially dynamic settings. 
\end{remark}


\begin{reptheorem}{nonTrivFDErr}[Non-Trivial Fully Dynamic Error Correction]
Our fully dynamic algorithm can perform non-trivial error correction, reducing $\eta_e$ to $\tilde{\eta}_e$ (minimum set edit distance from any point of predicted updates to  real updates) using $O(\tilde{\eta}_e)$ overhead cost per update.
\label{thm:nonTrivFDErr}
\end{reptheorem}

\begin{remark}
     This reduces the error to the minimum set edit distance, which can be used for both partially dynamic and fully dynamic settings using our fully dynamic algorithm. 
\end{remark}

\end{enumerate}



At its core, our algorithms can also be used as a batch update version of ES trees~\cite{even-shiloach} for incremental, decremental, and fully dynamic settings, which may be of independent interest. Notably, this simple extension to batch updates is sufficient to give a prediction algorithm satisfying \textit{robustness, consistency, and smooth degradation}. Thus, our framework can potentially be generalized for other dynamic graph problems by considering their corresponding batch-update version.

\subsubsection*{Outline.}
We describe the notations and definitions used in the paper in~\Cref{sec:preliminaries}, along with a brief description of the state-of-the-art algorithms. \Cref{sec:incremental} presents our incremental algorithm and 
\Cref{sec:decremental} presents our decremental algorithm, with their space efficient versions deferred to \Cref{sec:storage} and \Cref{sec:storage-decremental-fullyDynamic}.
\Cref{sec:fully-dynamic} combines these into a fully dynamic framework.
The approach to error correction, improving the performance of our algorithms, is described in \Cref{sec:error-correction}.
\Cref{sec:conclusion} concludes with directions for future work.


\section{Preliminaries}
\label{sec:preliminaries}


Let $G_0=(V,E_0)$ denote the initial directed graph, where $n=|V|$, $m_0=|E_0|$ and a source $s\in V$. The graph
is subject to an online sequence of $u$ edge updates
$U=\{\langle e_1,t_1\rangle,\langle e_2,t_2 \rangle,\dots,\langle e_u,t_u\rangle\}$, where each $e_i$ is an edge which is inserted ($t_i=+$)
or deleted ($t_i=-$). For incremental or decremental algorithms we assume $u=m$ and $t_i=+$ or $-$, which starts or ends at an empty graph. For the fully dynamic case $u$ can potentially be much larger than $m$, as same edges may be potentially added or removed multiple times. 
We denote by $G_i=(V,E_i)$ the graph obtained
after applying updates $\{\langle e_1,t_1\rangle,\cdots,\langle e_i,t_i \rangle\}$ to $G_0$.
We denote a BFS tree of graph $G_i$ from $s$ as $T_i$. 
For each $v\in V$, we define the following: 
\begin{itemize}
    \item $\texttt{In}_i[v]=\{u\in V:(u,v)\in E_i\}$, the set of in-neighbours of $v$ in $G_i$,
    \item $\texttt{Out}_i[v]=\{w\in V:(v,w)\in E_i\}$, the set of out-neighbours of $v$  in $G_i$,
    \item $\dg[i]{v}=|\texttt{In}_i[v]|+|\texttt{Out}_i[v]|$, the degree of $v$  in $G_i$.
    \item $\Level[i]{v}$: the distance from $s$ to $v$ in $T_i$, for unreachable vertices assume $\Level[i]{v}=n$.
    \item $\Par[i]{v}$: the parent of $v$ in $T_i$, with $\Par[i]{x}=\textsc{null}$ for $x=s$ and $x\notin T_i$.
     \item $\texttt{UP}_i[v]=\{u\in \texttt{In}_i[v]: \Level[i]{u}=\Level[i]{v}-1\}$, the \textit{upper parents}, or $\texttt{In}_i[v]$ at $\Level[i]{v}-1$ in $T_i$.
\end{itemize}
When the index $i$ corresponding to $G_i$ is clear from context, we drop the subscript $i$. However, for simplicity we use $\dg{v}$ to indicate the maximum degree of $v$ during the corresponding update sequence. 
This set of upper parents is only used by the decremental and fully dynamic algorithms to efficiently locate a replacement parent when the current parent edge is deleted.



\subsection{Prediction Model and Error Measures}

We assume an additional input, a \emph{predicted} update sequence $\hat{U}=\{\langle\hat{e}_1,\hat{t}_1\rangle,\langle\hat{e}_2,\hat{t}_2\rangle,\dots,\langle\hat{e}_u,\hat{t}_u\rangle\}$ provided by an external oracle (e.g., a learned model), which is available in full before any real update arrives. Again, we assume $u=m$ for incremental and decremental cases. Thus, it can be preprocessed offline to improve the processing of real updates. Let $\hat{G}_i=(V,\hat{E}_i)$ denote the graph obtained by
applying $\{\langle\hat{e}_1,\hat{t}_1\rangle,\cdots ,\langle\hat{e}_i,\hat{t}_i\rangle\}$ to $G_0$, and let $\hat{T}_i$ be a BFS
tree of $\hat{G}_i$ rooted at $s$. We store the corresponding arrays
$\hLevel[i]{\cdot}$ and $\hPar[i]{\cdot}$ for all $\hat{G}_i$. For the decremental and fully dynamic algorithms, we also store $\hat{\texttt{UP}}_i[\cdot]$ for each $\hat{G}_i$. The predicted sequence $\hat{U}$ need not agree with $U$ due to inaccurate predictions. 
Let the maximum exact prefix match of $U$ and $\hat{U}$ be of $i$ edges, i.e., $\hat{e}_k=e_k,\hat{t}_k=t_k$ for $k\leq i$ and at least $\hat{e}_{i+1}\neq e_{i+1}$ or $\hat{t}_{i+1}\neq t_{i+1}$. 
We thus define the following error measures for the current sequence $U=\{\langle e_1,t_1\rangle,\cdots,\langle e_j,t_j\rangle\}$  for $j \geq i$ updates with respect to predicted sequence $\hat{U}$.

\begin{definition}[Edge Prediction Error $\eta_e$]
\label{def:eta-e}
The number of real edge updates that were not exactly predicted, starting from the first point of error $\eta_e = j - i$
\end{definition}
\begin{remark}
This simple error measure is particularly significant for the partially dynamic updates, where a single incorrect update makes the entire remaining sequence unusable. 
\end{remark}

\begin{remark}
For fully dynamic case, we further exclude the edges which are both inserted and deleted in remaining real updates $\{\langle e_{i+1},t_{i+1}\rangle,\cdots,\langle e_j,t_j\rangle\}$.
\end{remark}

\begin{remark}
The simple $\eta_e$ measure is a coarse upper bound on the length of the divergent suffix. Error correction (Section~\ref{sec:error-correction}) reduces it using Set Matching, Prefix Set, and Set Edit Distance based error measures. 
\end{remark}

\begin{definition}[Vertex Prediction Error $\eta_v$]
\label{def:eta-v}
The sum of degrees $\dg{v}$ of vertices whose predicted level $\hat{\texttt{L}}_i[v]$ at last correct step
$i$ differs from the actual level $\texttt{L}_j[v]$ at step $j$
$$\eta_v \;=\; \sum_{\substack{v \in V :\hLevel[i]{v} \neq \Level[j]{v}}} \dg{v}$$
\end{definition}
\begin{remark}
Our erroneous prediction considers distance (instead of parent, etc.), making the error measure purely dependent on the graph and oblivious to the algorithm. 
\end{remark}

\begin{remark}
The degree of erroneously predicted vertex impacts the repair by affecting more incident edges, making $\eta_v$ a cost-sensitive error measure instead of a vertex-count error.
\end{remark}

\begin{definition}[Weighted Vertex Prediction Error $\eta^*_v$]
\label{def:eta-star}
The weighted sum of degrees $\dg{v}$ of vertices, with weight equivalent to error in prediction of level $|\hat{\texttt{L}}_i[v]-\texttt{L}_j[v]|$ from the last correct step $i$ to the current step $j$.

$$\eta^*_v \;=\; \sum_{v \in V} \dg v\,\bigl|\Level[j]{v} - \hLevel[i]{v}\bigr|$$
\end{definition}
\begin{remark}
The size of the distance error, as a weight for erroneously predicted vertices, captures the additional impact of the error,  making it a suitable estimate for the update cost.
\end{remark}

\begin{remark}
    Note that based on the above definitions, $\eta_e,\eta_v\leq m$ and $\eta^*_v\leq mn$.
\end{remark}

\subsubsection*{Why we use offline preprocessing.}
The predicted update sequence is available before the real updates arrive, so we can preprocess it offline and build the corresponding snapshots and auxiliary data structures in advance. This does not change the fact that the problem is online; rather, it shifts work away from the critical update path and allows the algorithm to exploit accurate predictions whenever they are available. Thus, even if the preprocessing cost is high, the offline phase remains justified because it reduces online latency and preserves correctness across all prediction qualities.

\subsection{State-of-the-art}
\label{sec:classical}

We briefly summarize the classical ES-tree approach~\cite{even-shiloach}, which
our algorithms extend. 
For completeness, we present the pseudocode for the classical single-edge
incremental and decremental BFS algorithms~\cite{even-shiloach} matching
our notation. 
These serve as a direct point of comparison for our batch
procedures \textsc{batchInsertEdge} and \textsc{batchDeleteEdge} presented in
Sections~\ref{sec:incremental} and~\ref{sec:decremental} respectively.

\subsubsection*{Incremental BFS.}
When an edge $(u,v)$ is inserted, if $\Level{u}+1<\Level{v}$, then the new
edge offers a shorter path to $v$. The algorithm updates $\Par{v}\gets u$ and
$\Level{v}\gets\Level{u}+1$, then propagates potential level decreases to
the out-neigbours of $v$ using a level-by-level queue $\LRq{\cdot}$, processing vertices in
non-decreasing order of their new levels. Hence, whenever the level of a vertex $v$ changes (decreases), we require $O(deg(v))$ time to process its neighbours. Since the level of a vertex can change $O(n)$ times throughout the algorithm, the total time required is $\sum_{v\in V}\dg{v}\times O(n)=O(mn)$.

\begin{algorithm}[htbp]
\caption{\textsc{InsertEdge}$(e_i=(u,v))$}
\label{alg:classicalInsertEdge}
\begin{multicols}{2}
    add $v$ to $\Out{u}$\;

    \BlankLine
    \lIf{$\Level{v} \le \Level{u}+1$}{
        \Return
    }
    \BlankLine
    Initialize empty lists $\LRq{\ell}$ for all levels $\ell$\;
    $\Par{v} \leftarrow u$\;
    $\Level{v} \leftarrow \Level{u}+1$\;
    add $v$ to $\LRq{\Level{v}}$\;

    \For{$l \leftarrow \Level{v}$ \KwTo $n$}{
        \While{$\LRq{l}$ not empty}{
            pop $y$ from $\LRq{l}$\;
            \ForEach{$z \in \Out{y}$}{
                
                \If{$\Level{z} > \Level{y}+1$}{
                    $\Par{z} \leftarrow y$\;
                    $\Level{z} \leftarrow \Level{y}+1$\;
                    add $z$ to $\LRq{\Level{z}}$\;
                }
            }
        }
    }
\end{multicols}
\end{algorithm}

\subsubsection*{Decremental BFS.}
When an edge $(u,v)$ is deleted, if $(u,v)$ was $v$'s parent edge (i.e.,
$u=\Par{v}$), then $v$ must find a replacement parent. The algorithm searches
$\texttt{Up}[v]$ for an incoming neighbour at level $\Level{v}-1$; if one exists it becomes the
new parent of $v$ and no further change is needed. Otherwise, $v$'s level
must increase, which may in turn invalidate the parent edges of $v$'s
children, triggering a recursive repair. Whenever the level of a vertex $v$ changes (increases), it updates $\texttt{Up}[v]$ and $\texttt{Up}[w]$ for each $(v,w)\in E$, requiring $O(\dg{v})$ time. Vertices are processed level-by-level
in non-decreasing order using queue \LLq{\cdot}, ensuring that by the time a vertex is examined, all
vertices at smaller levels have already been finalized. Again, as the level of a vertex can change $O(n)$ times throughout the algorithm, the total time required is $\sum_{v\in V}\dg{v}\times O(n)=O(mn)$.

\begin{remark}
    We have described the version of the decremental algorithm that matches our fully dynamic batch updates using $UP[\cdot]$ lists. A similar algorithm using pointer access to the adjacency list instead is possible (see \Cref{sec:pointerDec} for details).
\end{remark}


\begin{algorithm}[htbp]
\caption{\textsc{DeleteEdge}$(e_i=(u,v))$}
\label{alg:classicalDeleteEdge}
\begin{multicols}{2}
    remove $v$ from $\Out{u}$\;
    remove $u$ from $\In{v}$\;

    Initialize empty lists $\LLq{\ell}$ for all levels $\ell$\;

    \If{$u \in \UP{v}$}{
        remove $u$ from $\UP{v}$\;

        \uIf{$\UP{v} = \emptyset$}{
            $\Par{v} \leftarrow $ \textsc{null}\;
            add $v$ to $\LLq{\Level{v}}$\;
        }
        \ElseIf{$\Par{v} = u$}{
            $\Par{v} \leftarrow$ $\UP{v}[0]$\;
            \Return\;
        }
    }
 \columnbreak
 
    \For{$\ell \leftarrow \Level{v}$ \KwTo $n$}{

    \While{$\LLq{\ell}$ not empty}{
        pop $x$ from $\LLq{\ell}$\;
        \lIf{$\UP{x}\neq \emptyset$}{
            $\Par{x} \leftarrow$ $\UP{x}[0]$
        }\Else{

        $\Level{x} \leftarrow \Level{x}+1$\;
        add $x$ to $\LLq{\Level{x}}$\;
        $\UP{x}\gets \emptyset$\;
        \ForEach{$p\in \In{x}$}{
        \lIf{$\Level{p}=\Level{x}-1$}{
        add $p$ to $\UP{x}$
        }
        }
        \ForEach{$y \in \Out{x}$}{
            \If{$x \in \UP{y}$}{
            remove $x$ from $\UP{y}$\;
            \uIf{$\UP{y}=\emptyset$}{
             $\Par{y} \leftarrow $ \textsc{null}\;
             
                add $y$ to $\LLq{\Level{y}}$\;
            }
            \ElseIf{$\Par{y}=x$}{
                $\Par{y} \leftarrow$ $\UP{y}[0]$\;
            }
            }
        }
        }
    }
    }

\end{multicols}
\end{algorithm}

\newpage

\section{Incremental BFS with Predictions}
\label{sec:incremental}

We consider the incremental setting where $G_0(V,E_0)$ undergoes a sequence of edge insertions. 
We are given a predicted insertion sequence
$\hat{U}$ and a real insertion
sequence $U$, revealed online one edge at a time where all the $\hat{t}'s$ are $+$.
At step $j$, edge $e_j$ is inserted into the graph, producing $G_j=(V,E_j)$.

\subsection{Algorithmic Framework}

The algorithm proceeds in two phases: an offline preprocessing phase over the
predicted sequence, and an online phase that handles real updates. See \Cref{fig:inc-example} for an example. 
The full pseudocode is given in Algorithm~\ref{alg:batchInsertEdge}, with differences from \Cref{alg:classicalInsertEdge} highlighted in \textcolor{blue}{blue}. 

\subsubsection*{Preprocessing.}
We simulate the predicted insertions $\hat{e}_1,\dots,\hat{e}_m$ in order,
computing the sequence of predicted BFS trees $\hat{T}_0,\hat{T}_1,\dots,\hat{T}_m$.
For each $i$, we store the arrays $\hLevel[i]{\cdot}$ and $\hPar[i]{\cdot}$
corresponding to $\hat{T}_i$. This phase runs entirely offline before any real
update arrives.

\subsubsection*{Online Phase.}
We process the real insertions $e_1,\dots,e_m$ one by one, maintaining the
adjacency list $\Out{\cdot}$ of the current graph and an integer
$\mathit{i}$, initialised to $0$. Intuitively,
$\mathit{i}$ tracks the latest step at which the real and predicted prefix sequences match, so that $\hat{T}_{\mathit{i}}$ is a valid BFS tree of
$G_{\mathit{i}}$.

At step $j$, after inserting $e_j$ into $\Out{\cdot}$, we distinguish two
cases.

\medskip\noindent
\textbf{Case 1 (sequences agree).}
If the sequences $\{e_1,\dots,e_j\}=\{\hat{e}_1,\dots,\hat{e}_j\}$ or $e_i=\hat{e}_i$ for $i\in[i,j]$, then $\hat{T}_j$ is a
valid BFS tree of $G_j$. We answer the query directly from the stored arrays
$\hLevel[j]{\cdot}$ and $\hPar[j]{\cdot}$, and set $\mathit{i}\leftarrow j$.

\medskip\noindent
\textbf{Case 2 (sequences diverge).}
Otherwise, the two sequences have diverged at or before step $j$. We consider the
stored snapshot $\hat{T}_{\mathit{i}}$, which is a valid BFS tree of
$G_{\mathit{i}}$, and call \textsc{batchInsertEdge} (\Cref{alg:batchInsertEdge}) with the unprocessed
real edges $e_{\mathit{i}+1},\dots,e_j$ to repair the BFS tree. After 
answering the query from the updated arrays $\Level{\cdot}$ and $\Par{\cdot}$,
we roll back all temporary changes to revert the maintained data structures
to $\hat{T}_{\mathit{i}}$.

\Cref{alg:batchInsertEdge} iteratively processes all the edges to be inserted. Similar to \Cref{alg:classicalInsertEdge}, if insertion of an edge $(u,v)$ results in a shorter path for $v$, it is added to \LRq{$\Level{v}$} after its level and parent are updated. However, before updating its level, we remove $v$ (if present) from \LRq{$\Level{v}$} to avoid duplicates. To process \LRq{\cdot} in the increasing order of levels, we compute the minimum level $\ell^*$ for which \LRq{\ell^*} is not empty and process levels incrementally from $\ell^*$ to $n$.



    

        
        

\begin{figure}[htbp]
    \centering
\includegraphics[width=1.1\linewidth]{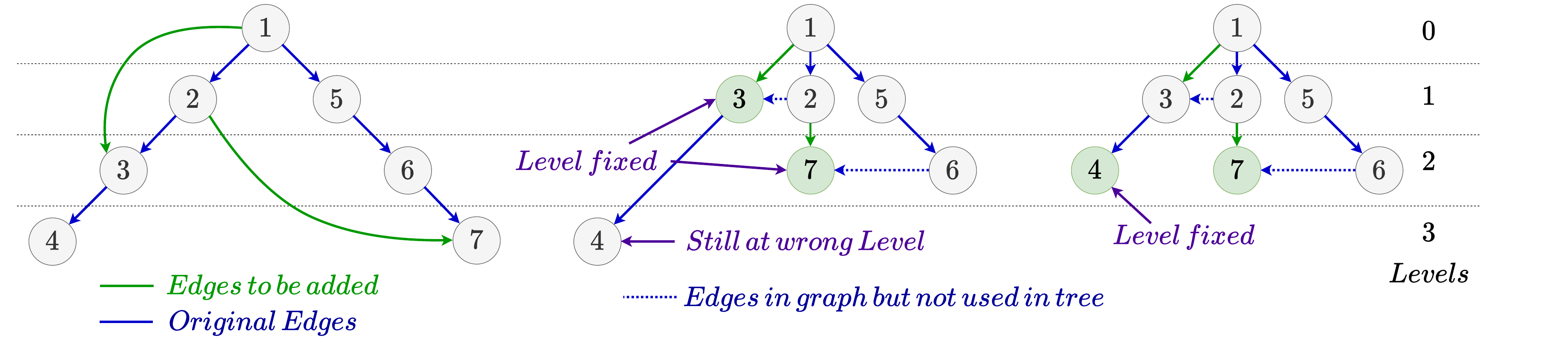}
    \caption{Incremental BFS example. During initialization, nodes 3 and 7 move to \LRq{1} and \LRq{2} respectively, and their levels are also fixed. The propagation pops node 3, looks at its children, fixes node 4 from level 3 to level 2, and adds it to \LRq{2}. Then at the next level, both 4 and 7 are popped, and propagation ends.}
    \label{fig:inc-example}
\end{figure}


\begin{algorithm}[htbp]
\caption{\textsc{batchInsertEdge}$(U_{i..j}=\{\langle e_{i+1},t_{i+1}\rangle,\ldots,\langle e_j,t_j\rangle\})$}
\label{alg:batchInsertEdge}
\begin{multicols}{2}

Initialize empty lists $\LRq{\ell}$ for all levels $\ell$\;

{
\color{blue}
$\ell^* \leftarrow n$\;
\ForEach{$((u,v),+) \in U_{i..j}$}{
    {\color{black}
    add $v$ to $\Out{u}$\;    
    \If{$\Level{v} > \Level{u}+1$}{
{\color{blue}
remove $v$ from $\LRq{\Level{v}}$ (if any)\; }
        $\Par{v} \leftarrow u$\;
        $\Level{v} \leftarrow \Level{u}+1$\;
        add $v$ to $\LRq{\Level{v}}$\;  
        {\color{blue}
        $\ell^* \leftarrow \min(\ell^*,\, \Level{v})$\;}
     }     
     }
}
}

\For{$\ell \leftarrow \textcolor{blue}{\ell^*}$ \KwTo $n$}{
    \While{$\LRq{\ell}$ not empty}{
        pop $y$ from $\LRq{\ell}$\;
        
        \ForEach{$z \in \Out{y}$}{
            \If{$\Level{z} > \Level{y}+1$}{
               
            \textcolor{blue}{
            remove $z$ from $\LRq{\Level{z}}$ (if any)\; 
            }    
                $\Par{z} \leftarrow y$\;
                $\Level{z} \leftarrow \Level{y}+1$\;
                add $z$ to $\LRq{\Level{z}}$\;
             }
        }
    }
}
\end{multicols}
\end{algorithm}

\newpage

\subsection{Correctness}

\begin{lemma}
\label{lem:incremental-correctness}
After \textsc{batchInsertEdge} processes all levels up to $\ell$, the arrays
$\Level{\cdot}$ and $\Par{\cdot}$ are correct for every vertex whose distance
from $s$ in $G_j$ is at most $\ell$.
\end{lemma}

\begin{proof}
We proceed by induction on $\ell$, starting from $\ell^*$, the smallest level
at which any vertex is affected by the batch insertion.

\medskip\noindent
\textit{Base case} ($\ell=\ell^*$). Before the batch, the BFS tree is correct
for all levels below $\ell^*$. During initialization, every vertex that can
reach distance $\ell^*$ via a newly inserted edge is identified and assigned
the correct level and parent. Hence, all vertices at a distance at most $\ell^*$
are correct after this step.

\medskip\noindent
\textit{Inductive step}. Suppose all vertices at a distance at most $\ell$ are
correct. Consider any vertex $v$ with true distance $\ell+1$ in $G_j$. Its
shortest-path parent $p$ satisfies $\Level{p}=\ell$ and $(p,v)\in E_j$. When
$p$ is processed at level $\ell$, the edge $(p,v)$ is relaxed, setting
$\Level{v}\leftarrow\ell+1$ and $\Par{v}\leftarrow p$. No shorter distance is
possible for $v$, since all vertices at levels below $\ell+1$ are already
finalized by the induction hypothesis. Hence, all vertices at a distance of at most
$\ell+1$ are correct after processing level $\ell+1$.
\end{proof}

\subsection{Analysis}
\subsubsection*{Preprocessing.} We compute the required data structures using the classical Incremental BFS algorithm, maintaining a copy of $T_i$ and its associated structures after every update. This requires $O(mn)$ time and space.

\subsubsection*{Update \textsc{batchInsertEdge}.}
Each edge in the batch is processed once during the initialization phase,
contributing $O(\eta_e)$ work in total. Let $A\subseteq V$ denote the set of
vertices whose level strictly decreases as a result of the batch insertion.
For each $v\in A$, we scan all edges in $\Out{v}$ once. The total running time
is therefore $O\!\left(\eta_e + \sum_{v \in A} \dg{v}\right) = O(\eta_e + \eta_v),
$
where the equality follows from Definition~\ref{def:eta-v}, since $A$ is
exactly the set of vertices whose level in $G_j$ differs from their predicted
level at step $\mathit{i}$. The space required is also bounded by $O(mn)$ due to previously built data structures.



\begin{theorem}
    Given a graph $G(V,E)$ having $n$ vertices and $m$ edges undergoing edge insertions with a predicted sequence of updates, the BFS tree can be maintained in worst case $O(\eta_e+\eta_v)$ time using $O(mn)$ preprocessing time and space. 
\end{theorem}

\subsubsection*{Bounding total time.} The BFS tree $T_j$ can also be obtained incrementally from $T_{j-1}$ by inserting only $e_j$. The affected vertices in this one-step transition form a subset of those in the batch transition from $\hat{T}_{\mathit{i}}$ to $T_j$, so it is also $O(\eta_e+\eta_v)$. We describe the algorithm for batch update only because it extends naturally to the fully dynamic setting. 
Now, by building $T_j$ from $T_{j-1}$ using a single edge update, we can also bound the total update time. Since each vertex is processed at most once for every change in its level, the resulting total time is $O(\eta_e+\eta_v^*)$.

\repeattheorem{incBFS}

\newpage
\section{Decremental BFS with Predictions}
\label{sec:decremental}

We consider the decremental setting in which $G_0=(V,E_0)$ undergoes a
sequence of edge deletions.
%
%
We are given a predicted deletion sequence
$\hat{U}$ and a real deletion
sequence $U$, revealed online one edge at a time where all the $\hat{t}'s$ are $-$.
At step $j$, edge $e_j$ is deleted from the graph, producing $G_j=(V,E_j)$. Our goal is to exploit the agreement between
$\hat{U}$ and $U$ to accelerate BFS-tree maintenance.

\subsection{Algorithmic Framework}

The algorithm proceeds in two phases: an offline preprocessing phase over the
predicted sequence, and an online phase that handles real updates. See \Cref{fig:dec-example} for an example. 
The full pseudocode is given in Algorithm~\ref{alg:batchDeleteEdge}, with differences from \Cref{alg:classicalDeleteEdge} highlighted in \textcolor{blue}{blue}. 

\subsubsection*{Preprocessing.}
We simulate the predicted deletions $\hat{e}_1,\dots,\hat{e}_m$ in order,
computing the sequence of predicted BFS trees $\hat{T}_0,\hat{T}_1,\dots,\hat{T}_m$.
For each $j$, we store the arrays $\hLevel[j]{\cdot}$, $\hPar[j]{\cdot}$, and
the upper-parent lists $\UP{\cdot}$ corresponding to $\hat{T}_j$. This phase
runs entirely offline before any real update arrives.

\subsubsection*{Online Phase.}
We process the real deletions $e_1,\dots,e_m$ one by one, maintaining the
adjacency lists $\In{\cdot}$ and $\Out{\cdot}$ of the current graph and an
integer $\mathit{i}$, initialised to $0$. As in the
incremental setting, $\mathit{i}$ tracks the latest step at which the real
and predicted prefix sequences match, so that $\hat{T}_{\mathit{i}}$ is a
valid BFS tree of $G_{\mathit{i}}$.

At step $j$, after deleting $e_j$ from $\In{\cdot}$ and $\Out{\cdot}$, we
distinguish between the two cases.

\medskip\noindent
\textbf{Case 1 (sequences agree).}
If $\{e_1,\dots,e_j\}=\{\hat{e}_1,\dots,\hat{e}_j\}$ or $e_i=\hat{e}_i$ for $i\in [1,j]$, then $\hat{T}_j$ is a
valid BFS tree of $G_j$. We answer the query directly from the stored arrays
$\hLevel[j]{\cdot}$ and $\hPar[j]{\cdot}$, and set $\mathit{i}\leftarrow j$.

\medskip\noindent
\textbf{Case 2 (sequences diverge).}
Otherwise, the two sequences have diverged at or before step $j$. We consider the
stored snapshot $\hat{T}_{\mathit{i}}$, which is a valid BFS tree of
$G_{\mathit{i}}$, together with its arrays $\hLevel[\mathit{i}]{\cdot}$,
$\hPar[\mathit{i}]{\cdot}$, and upper-parent lists $\UP{\cdot}$. We then
call \textsc{batchDeleteEdge} (\Cref{alg:batchDeleteEdge}) with the unprocessed real edges
$e_{\mathit{i}+1},\dots,e_j$ to repair the BFS tree. After answering the
query from the updated arrays $\Level{\cdot}$ and $\Par{\cdot}$, we roll back
all temporary changes to revert the maintained data structures to
$\hat{T}_{\mathit{i}}$.

\begin{figure}[htbp]
    \centering
    \includegraphics[width=1.1\linewidth]{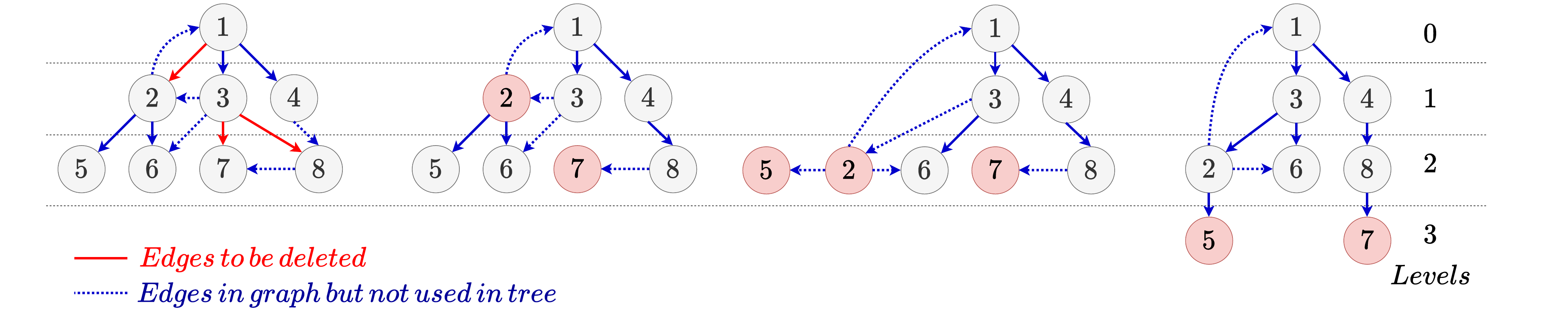}
    \caption{Decremental BFS example. During initialization, nodes 2, 7, and 8 lose their current parents, but node 8 already has an extra node 4 in its \UP{} so it reassigns to that directly, while 2 and 7 are pushed to \LLq{1} and \LLq{2}, respectively. During propagation, node 2 falls down while its children have it removed from their \UP{}, 5 is pushed to \LLq{2} while 6 directly gets 3 from its \UP{}. Then at the next level, 2 gets 3 as a parent but 5 and 7 fall down. Then at last level 5 and 7 get new parents.}
    \label{fig:dec-example}
\end{figure}

\Cref{alg:batchDeleteEdge} iteratively processes all the edges to be deleted. Similar to \Cref{alg:classicalDeleteEdge}, deletion of a tree edge $(u,v)$ results in finding an alternate parent from \UP{v} (if possible) to maintain its level, else $v$ is added to \LLq{$\Level{v}$}. Note that duplicates are not possible because a vertex is only added to \LLq{\cdot} when its parent becomes \textsc{null}, which is possible only once. To process \LLq{\cdot} in the increasing order of levels, we compute the minimum level $\ell^*$ for which \LLq{\ell^*} is not empty and process levels incrementally from $\ell^*$ to $n$.

\noindent
\begin{algorithm}[htbp]
\caption{\textsc{batchDeleteEdge}$(U_{i..j}=\{\langle e_{i+1},t_{i+1}\rangle,\ldots,\langle e_j,t_j\rangle\})$}
\label{alg:batchDeleteEdge}
\begin{multicols}{2}

Initialize empty lists $\LLq{\ell}$ for all levels $\ell$\;
{
\color{blue}
$\ell^* \leftarrow n$\;
\ForEach{$((u,v),-) \in U_{i..j}$}{
{\color{black}
    remove $v$ from $\Out{u}$\;
    remove $u$ from $\In{v}$\;

    \If{$u \in \UP{v}$}{
        remove $u$ from $\UP{v}$\;
        \uIf{$\UP{v}=\emptyset$}{
         $\Par{v} \leftarrow $ \textsc{null}\;
            add $v$ to $\LLq{\Level{v}}$\;
            \textcolor{blue}{
            $\ell^* \leftarrow \min(\ell^*,\,\Level{v})$\;}
        }
        \ElseIf{$\Par{v}=u$}{
            $\Par{v} \leftarrow$ $\UP{v}[0]$\;
        }
    }
    }
}
}
\columnbreak
\For{$\ell \leftarrow \textcolor{blue}{\ell^*}$ \KwTo $n$}{
  \While{$\LLq{\ell}$ not empty}{
        pop $x$ from $\LLq{\ell}$\;
        \lIf{$\UP{x}\neq \emptyset$}{
            $\Par{x} \leftarrow$ $\UP{x}[0]$
        }\Else{

        $\Level{x} \leftarrow \Level{x}+1$\;
        add $x$ to $\LLq{\Level{x}}$\;
        $\UP{x}\gets \emptyset$\;
        \ForEach{$p\in \In{x}$}{
        \lIf{$\Level{p}=\Level{x}-1$}{
        add $p$ to $\UP{x}$
        }
        }
        \ForEach{$y \in \Out{x}$}{
            \If{$x \in \UP{y}$}{
            remove $x$ from $\UP{y}$\;
            \uIf{$\UP{y}=\emptyset$}{
             $\Par{y} \leftarrow $ \textsc{null}\;
                add $y$ to $\LLq{\Level{y}}$\;
            }
            \ElseIf{$\Par{y}=x$}{
                $\Par{y} \leftarrow$ $\UP{y}[0]$\;
            }
            }
        }
        }
    }
}
\end{multicols}
\end{algorithm}

\subsection{Correctness}

\begin{lemma}
\label{lem:decremental-correctness}
After \textsc{batchDeleteEdge} processes all levels up to $\ell$, the arrays
$\Level{\cdot}$ and $\Par{\cdot}$ are correct for every vertex whose distance
from $s$ in $G_j$ is at most $\ell$.
\end{lemma}

\begin{proof}
We proceed by induction on $\ell$, starting from $\ell^*$, the smallest level
at which any vertex is affected by the batch deletion.

\medskip\noindent
\textit{Base case} ($\ell=\ell^*$). Before processing the batch, the BFS tree
is correct up to level $\ell^*-1$. During the initialization pass, every
vertex that directly lost all candidate parents at level $\ell^*-1$ is
enqueued in $\LLq{\ell^*}$ and repaired; vertices that still retain at least
one candidate in $\UP{\cdot}$ are immediately assigned a valid parent. Hence,
all vertices at a distance of at most $\ell^*$ are correct after initialization.

\medskip\noindent
\textit{Inductive step}. Assume all vertices at a distance at most $\ell$ are
correct. When we process level $\ell+1$, each vertex $v$ dequeued from
$\LLq{\ell+1}$ either finds a valid parent in $\UP{v}$ and is assigned one,
or has $\UP{v}=\emptyset$ and is moved to level $\ell+2$ and re-enqueued.
In the latter case, the children of $v$ whose sole candidate parent was $v$
are updated accordingly. Since all vertices at levels at most $\ell$ are
already correct by the induction hypothesis, correctness extends to level
$\ell+1$, completing the induction.
\end{proof}

\newpage

\subsection{Analysis}

\subsubsection*{Preprocessing.} We compute the required data structures using the classical Decremental BFS algorithm, maintaining a copy of $T_i$ and its associated structures after every update. This requires $O(mn)$ time and space. However, we additionally maintain a copy of $\texttt{UP}_i[v]$ for each $\hat{G}_i$ requiring $O(m^2)$ time and space. 

\subsubsection*{Update \textsc{batchDeleteEdge}.}
Each edge in the batch is processed once during the initialization phase,
contributing $O(\eta_e)$ work in total. Let $A\subseteq V$ denote the set of
vertices whose level strictly increases as a result of the batch deletion. For
each $v\in A$, we scan all edges in $\In{v}$ and $\Out{v}$ once per level
increase of $v$. The total running time is therefore
\[
    O\!\left(\eta_e + \sum_{v \in A} \dg v\cdot|\Level[j]{v} - \hLevel[\mathit{i}]{v}|\right)
    \;\subseteq\;
    O\!\left(\eta_e + \sum_{v \in V} \dg v\cdot|\Level[j]{v} - \hLevel[\mathit{i}]{v}|\right)
    = O(\eta_e + \eta^*_v),
\]
where the last equality follows from Definition~\ref{def:eta-star}. The space required is also bound by the data structures built during preprocessing. 

\begin{remark}
In the worst case each vertex's level may increase by $O(n)$, so
$\eta^*_v = O(n\sum_{v\in V}\dg v) = O(nm)$.
When $\eta_e+\eta^*_v$ exceeds $O(m)$, we fall back to rebuilding the BFS
tree from scratch in $O(m)$ time (similar approach was previously used by \cite{hanauer2020fully}).
\end{remark}


\begin{theorem}
    Given a graph $G(V,E)$ having $n$ vertices and $m$ edges undergoing edge deletions with a predicted sequence of updates, the BFS tree can be maintained in worst case $O(\min\{m,\eta_e+\eta^*_v\})$ time using $O(m^2)$ preprocessing time and space. 
\end{theorem}


\subsubsection*{Bounding total time.} Again, the BFS tree $T_j$ can also be obtained incrementally from $T_{j-1}$ by deleting only $e_j$. The affected vertices in this one-step transition form a subset of those in the batch transition from $\hat{T}_{\mathit{i}}$ to $T_j$, so it is also $O(\min\{m,\eta_e+\eta^*_v\})$. We describe the algorithm for batch update only because it extends naturally to the fully dynamic setting. 
Now, by building $T_j$ from $T_{j-1}$ using a single edge update, we can also bound the total update time. Since the adjacency list of each vertex is processed at most once for every change in its level, the resulting total time is $O(\eta_e+\eta_v^*)$.

\subsection{Pointer-Based Decremental Storage}
\label{sec:pointerDec}
In the decremental setting, we can replace $\UP{\cdot}$ by one cursor per vertex in $\In{v}$. Thus, instead of storing full upper-parent lists, we store $\texttt{ptr}[j]=(\texttt{ptr}_j[1],\dots,\texttt{ptr}_j[n])$, where $\texttt{ptr}_j[v]$ is the current scan position at step $j$.

After deleting $(u,v)$, if it is not the parent edge of $v$, only the adjacency lists change. Otherwise, we scan forward from $\texttt{ptr}[v]$ to find a valid parent at level $\Level{v}-1$; if none exists, we increase $\Level{v}$, enqueue $v$, and reset $\texttt{ptr}[v]$. This stores one pointer per vertex per snapshot, using $O(n)$ space per update and $O(mn)$ overall, and reproduces the parent-selection of $\UP{\cdot}$.

\begin{remark}
The $\UP{\cdot}$ formulation is kept for consistency with the fully dynamic framework. For decremental BFS, the pointer-based version is simpler and space-efficient, hence preferred.
\end{remark}

\repeattheorem{decBFS}

\newpage
\section{Fully Dynamic BFS with Predictions}
\label{sec:fully-dynamic}

We consider the fully dynamic setting in which $G_0=(V,E_0)$ undergoes a
sequence of edge insertions and deletions.
%
%
We are given a predicted update sequence
$\hat{U}$ and a real update sequence
$U$, revealed online one update at a time. At
step $j$, update $e_j$ is applied to the graph, producing $G_j=(V,E_j)$. Our goal is to exploit agreement between $\hat{U}$ and
$U$ to accelerate BFS-tree maintenance.

\subsection{Algorithmic Framework}

The algorithm proceeds in two phases: an offline preprocessing phase over the
predicted sequence, and an online phase that handles real updates. See \Cref{fig:dyn-example} for an example. 
The full procedure is given in Algorithm~\ref{alg:batchDynamicUpdate}, with differences from \Cref{alg:batchInsertEdge,alg:batchDeleteEdge} highlighted in \textcolor{blue}{blue}. 

\subsubsection*{Preprocessing.}
We simulate the predicted updates $\{\langle\hat{e}_1,\hat{t}_1\rangle,\langle\hat{e}_2,\hat{t}_2\rangle,\dots,\langle\hat{e}_u,\hat{t}_u\rangle\}$ in order,
computing the sequence of predicted BFS trees $\hat{T}_0,\hat{T}_1,\dots,\hat{T}_u$.
For each $j$, we store the arrays $\hLevel[j]{\cdot}$, $\hPar[j]{\cdot}$, and
the upper-parent lists $\UP{\cdot}$ corresponding to $\hat{T}_j$. This phase
runs entirely offline before any real update arrives.

\subsubsection*{Online Phase.}
We process the real updates $\{\langle e_1,t_1\rangle,\langle e_2,t_2 \rangle,\dots,\langle e_u,t_u\rangle\}$ one by one, maintaining the
adjacency lists $\In{\cdot}$ and $\Out{\cdot}$ of the current graph and an
integer $\mathit{i}$, initialised to $0$. As before,
$\mathit{i}$ tracks the latest step at which the real and predicted prefix sequences match, so that $\hat{T}_{\mathit{i}}$ is a valid BFS tree of
$G_{\mathit{i}}$.

At step $j$, we first apply update $(u,v,t)$ to the adjacency lists: if
$t=+$ we insert $(u,v)$ into $\In{\cdot}$ and $\Out{\cdot}$; if $t=-$ we
delete $(u,v)$ from $\In{\cdot}$ and $\Out{\cdot}$. We then have two
cases.

\medskip\noindent
\textbf{Case 1 (sequences agree).}
If $\{e_1,\dots,e_j\}=\{\hat{e}_1,\dots,\hat{e}_j\}$ or $e_i=\hat{e}_i$ for $i\in [1,j]$, then $\hat{T}_j$ is a
valid BFS tree of $G_j$. We answer the query directly from the stored arrays
$\hLevel[j]{\cdot}$ and $\hPar[j]{\cdot}$, and set $\mathit{i}\leftarrow j$.

\medskip\noindent
\textbf{Case 2 (sequences diverge).}
Otherwise, the two sequences have diverged at or before step $j$. We consider the
stored snapshot $\hat{T}_{\mathit{i}}$, together with its arrays
$\hLevel[\mathit{i}]{\cdot}$, $\hPar[\mathit{i}]{\cdot}$, and
upper-parent lists $\UP{\cdot}$. 
We call
\textsc{batchDynamicUpdate}$(U_{i..j})$
(\Cref{alg:batchDynamicUpdate}) to repair the
BFS tree. After answering the query from the updated arrays $\Level{\cdot}$ and
$\Par{\cdot}$, we roll back all temporary changes to revert the maintained data structures to $\hat{T}_{\mathit{i}}$.

\Cref{alg:batchDynamicUpdate} iteratively processes all the edges to be deleted and inserted. Similar to
\Cref{alg:batchDeleteEdge}, deletion of a tree edge $(u,v)$ either replaces the parent for $v$ to maintain its level,  or adds $v$ to \LLq{$\Level{v}$}. However, unlike \Cref{alg:batchInsertEdge}, even if insertion of an edge $(u,v)$ results in a shorter path for $v$, its level is not updated immediately, as $u$ might end up with a raised level due to deletions in its shortest path. Hence, we instead maintain a list of such edges in ${\cal L}^e_R[\Level{u}]$ (rather than vertices), which is processed (similar to \Cref{alg:batchInsertEdge}) only when the level of $u$ is finalized. Thus, when the level of a vertex $x$ is raised due to some deletion while processing \LLq{\ell}, all the edges $(x,\cdot)\in {\cal L}^e_R[\ell]$ are also moved to the next level ${\cal L}^e_R[\ell+1]$ to be processed after $x$ is finalized. Hence, upon processing $(y,z)\in {\cal L}^e_R[\ell]$, a vertex $z$ is added to \LRq{\ell+1} only when the complete \LLq{\ell} is processed. This also avoids duplicates in \LRq{\ell} as a vertex is added to the next level only when its level improves. Now, $x\in$\LRq{\ell} can be easily processed by adding all its outgoing edges to ${\cal L}^e_R[\ell]$.
Further, before a vertex $x$ gets a new parent while processing ${\cal L}^e_R[\ell]$, it is removed from \LLq{$\Level{x}$} to avoid raising an already finalized vertex. Finally, as in previous algorithms, to process \LLq{\cdot} in increasing order of levels, we compute the minimum level $\ell^*$ such that \LLq{\ell^*} is nonempty and process the levels incrementally from $\ell^*$ to $n$.

\begin{figure}[htbp]
    \centering
    \includegraphics[width=1.1\linewidth]{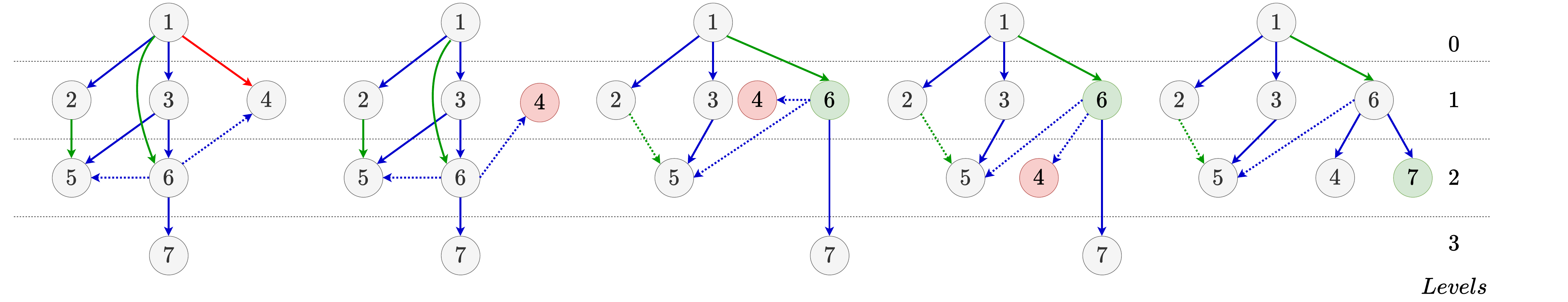}
    \caption{Fully dynamic BFS example. During initialization, node 4 loses its current parent and is pushed to \LLq{1}, meanwhile 2 is pushed in \UP{5}, and (1,6) is pushed to ${\cal L}^e_R[1]$. During propagation, 6 is lowered to level 1 and added to \LRq{1}. Then, node 4 falls down and goes into \LLq{2}, while 6 is popped from \LRq{1} and fixes level of both 4 and 7.}
    \label{fig:dyn-example}
\end{figure}

\begin{algorithm}[htbp]
\caption{\textsc{batchDynamicUpdate}$(U_{i..j}=\{\langle e_{i+1},t_{i+1}\rangle,\ldots,\langle e_j,t_j\rangle\})$}
\label{alg:batchDynamicUpdate}
\begin{multicols}{2}

Initialize empty lists $\LLq{\ell}, \LRq{\ell}, {\cal L}^e_R[\ell]$ for all $\ell$\;
$\ell^* \leftarrow n$\;
\BlankLine

\ForEach{$((u,v),-) \in U_{i..j}$}{
{\color{black}
    remove $v$ from $\Out{u}$\;
    remove $u$ from $\In{v}$\;

    \If{$u \in \UP{v}$}{
        remove $u$ from $\UP{v}$\;
        \uIf{$\UP{v}=\emptyset$}{
         $\Par{v} \leftarrow $ \textsc{null}\;
            add $v$ to $\LLq{\Level{v}}$\;
            
            $\ell^* \leftarrow \min(\ell^*,\,\Level{v})$\;
        }
        \ElseIf{$\Par{v}=u$}{
            $\Par{v} \leftarrow$ $\UP{v}[0]$\;
        }
    }
    }
}
\BlankLine

\ForEach{$((u,v),+) \in U_{i..j}$}
{
    add $v$ to $\Out{u}$\; 
    \textcolor{blue}{
    add $u$ to $\In{v}$\;
    }
    \uIf{$\Level{v}>\Level{u}+1$}{
\textcolor{blue}{        add $(u,v)$ to ${\cal L}^e_R[L[u]]$\;
      $\ell^* \leftarrow \min(\ell^*,\Level{u})$\;
       } }
    \textcolor{blue}{
    \ElseIf{$\Level{v}=\Level{u}+1$}{
        add $u$ to $\UP{v}$\;
        \If{$v \in \LLq{\Level{v}}$}{
    $\Par{v} \gets u$\;
    remove $v$ from $\LLq{\Level{v}}$\;
    }}
    }
}

\For{$\ell \leftarrow \ell^*$ \KwTo $n$}{
    \BlankLine
   \While{$\LLq{\ell}$ not empty}{
        pop $x$ from $\LLq{\ell}$\;
        \lIf{$\UP{x}\neq \emptyset$}{
            $\Par{x} \leftarrow$ $\UP{x}[0]$
        }\Else{
\textcolor{blue}{
move $(x,.)$ from ${\cal L}^e_R[\ell]$ to ${\cal L}^e_R[\ell+1]$             (if any)\;
}
        $\Level{x} \leftarrow \Level{x}+1$\;
        add $x$ to $\LLq{\Level{x}}$\;
        $\UP{x}\gets \emptyset$\;
        \ForEach{$p\in \In{x}$}{
        \lIf{$\Level{p}=\Level{x}-1$}{
        add $p$ to $\UP{x}$
        }
        }
        \ForEach{$y \in \Out{x}$}{
            \If{$x \in \UP{y}$}{
            remove $x$ from $\UP{y}$\;
            \uIf{$\UP{y}=\emptyset$}{
             $\Par{y} \leftarrow $ \textsc{null}\;
                add $y$ to $\LLq{\Level{y}}$\;
            }
            \ElseIf{$\Par{y}=x$}{
                $\Par{y} \leftarrow$ $\UP{y}[0]$
            }
            }
        }
        }
    }
    \BlankLine
    \While{$\LRq{\ell}$ not empty}{
        pop $y$ from $\LRq{\ell}$\;
        \ForEach{$z\in\Out{y}$}{
    \textcolor{blue}{
        add $(y,z)$ to ${\cal L}^e_R[\ell]$}        
        }}
    
    \textcolor{blue}{
    \While{${\cal L}^e_R[\ell]$ not empty}{
    pop $(y,z)$ from ${\cal L}^e_R[\ell]$\;
        \textcolor{black}{
       \uIf{$\Level{z}>\Level{y}+1$}{
            remove $z$ from \textcolor{blue}{ \LLq{\Level{z}}}
            (if any)\;
    $\Par{z}\leftarrow y$\; 
     $\Level{z}\leftarrow \Level{y}+1$\;
                \textcolor{blue}{
                $\UP{z}\leftarrow\{y\}$\;
                }
                add $z$ to $\LRq{\Level{z}}$\;
            }    
            \textcolor{blue}{
        \ElseIf{$\Level{z}=\Level{y}+1$}{
            add $y$ to $\UP{z}$\;
            \If{\(z\in\LLq{\ell+1}\)}{
                \(\Par{z}\leftarrow y\)\;
            remove \(z\) from \(\LLq{\ell+1}\)\;
            }
            }
            }
            }
    }
    }
}

\end{multicols}
\end{algorithm}

\newpage
\subsection{Correctness}

\begin{lemma}
\label{lem:fully-dynamic-correctness}
After \textsc{batchDynamicUpdate} sweeps levels from $\ell^*$ up to $\ell$,
the arrays $\Level{\cdot}$ and $\Par{\cdot}$ are correct for every vertex
whose distance from $s$ in $G_j$ is at most $\ell$.
\end{lemma}

\begin{proof}
We proceed by induction on $\ell$, starting from $\ell^*$, the smallest level
at which any vertex is affected by the batch update.

\medskip\noindent
\textit{Base case} ($\ell=\ell^*$). Before applying the batch $U_{i..j}$, 
the BFS tree is correct up to level
$\ell^*-1$. The initialization pass processes deletions before insertions,
enqueuing and fixing every vertex whose parent or candidate-parent set changed
directly at level $\ell^*$ due to deletion. For insertion, the potentially new parent edge is also enqueued for processing once the parent level is stabilized. Hence, all vertices at a distance of at most $\ell^*$ are correct after the initialization pass.

\medskip\noindent
\textit{Inductive step}. Assume all vertices at a distance at most $\ell$ are
correct. When processing level $\ell+1$, we first handle deletion-driven
level increases: each such vertex either finds a candidate parent at level
$\ell$ and is assigned one, or is raised to the next level and re-enqueued
with its children updated accordingly. Note that all vertices at level $\ell+1$ are now finalized. We then handle insertion-driven level
decreases, which are processed only when the source vertex's level is finalized. These relax outgoing edges and assign parent pointers for vertices
that newly obtain distance $\ell+2$. Since deletions are processed before
insertions at each level, and all vertices at levels at most $\ell$ are
already correct by the induction hypothesis, every vertex whose true distance is $\ell+1$ receives a correct parent and distance during this sweep. Hence, correctness extends to level $\ell+1$, completing the induction. 
\end{proof}

\subsection{Analysis}
\subsubsection*{Preprocessing.} We compute the required data structures using the trivial BFS algorithm, maintaining a copy of $T_i$ and $\texttt{UP}_i[v]$ for each $\hat{G}_i$, requiring $O(mu)$ time and space.

\subsubsection*{Update \textsc{batchDynamicUpdate}.}
Each edge in $U_{i..j}$ 
is processed once
during the initialization phase (and while processing ${\cal L}^e_R[\cdot]$), contributing $O(\eta_e)$ work in total. Note that edges $(x,\cdot)\in{\cal L}^e_R[\ell]$ are saved together using pointers such that moving them to ${\cal L}^e_R[\ell]$ takes $O(\dg{x})$ time when the level of $x$ increases. Let
$A_1\subseteq V$ denotes the set of vertices whose level strictly decreases due
to batch insertions, and let $A_2\subseteq V$ denote the set of vertices whose
level strictly increases due to batch deletions. Note that no vertex has both its level increased and decreased during the algorithm. For each $v\in A_1$ we scan
$\Out{v}$ once; for each $v\in A_2$ we scan $\In{v}$ and $\Out{v}$ once per
level increase. The total running time is therefore
\[
\begin{aligned}
&O\!\left(
    \eta_e
    + \sum_{v \in A_1} \dg v
    + \sum_{v \in A_2} \dg v\cdot|\Level[j]{v} - \hLevel[\mathit{i}]{v}|
\right) \\
\;\subseteq\;
&O\!\left(
    \eta_e
    + \sum_{v \in V} \dg v\cdot|\Level[j]{v} - \hLevel[\mathit{i}]{v}|
\right)
= O(\eta_e + \eta^*_v),
\end{aligned}
\]
where the inclusion holds because each term in the $A_1$ sum is bounded by the
corresponding weighted term (the level difference is at least $1$), and the
last equality follows from Definition~\ref{def:eta-star}. The required space is again bounded by the sizes of the data structures built during preprocessing.

\begin{remark}
In the worst case each vertex level may change by $O(n)$, so
$\eta^*_v = O(n\sum_{v\in V}\dg v) = O(nm)$.
When $\eta_e+\eta^*_v$ exceeds $O(m)$, we fall back to rebuilding the BFS
tree from scratch in $O(m)$ time  (similar approach was previously used by \cite{hanauer2020fully}).
\end{remark}


\repeattheorem{fdBFS}

\begin{remark}
Unlike the incremental and decremental settings, processing the single update
$e_j$ directly on $T_{j-1}$ is not necessarily better than the batch
transition from $\hat{T}_{\mathit{i}}$ to $T_j$. Insertions and deletions
within a batch may offset one another, resulting in fewer net level changes
than would arise from processing each update individually. The batch
formulation is therefore the natural choice for the fully dynamic setting, as otherwise the  worst case complexity is not bounded accordingly, losing the property of smooth degradation of the algorithm.
\end{remark}


\section{Space-Efficient Incremental BFS algorithm}
\label{sec:storage}

To use predictions, we must store the sequence of predicted BFS trees
$\hat{T}_0,\hat{T}_1,\dots,\hat{T}_m$. Naive storage requires $O(mn)$ space
if we keep $\hPar[j]{\cdot}$ and $\hLevel[j]{\cdot}$ for all $n$ vertices at
all $m$ steps. In the incremental setting, we can exploit the fact that BFS
levels only decrease to compress this history to $O(n^2)$ space.


\subsubsection*{Data structure.}
For each vertex $v\in V$, maintain a chronological list of the steps at which
its parent or level changes:
\[
    \mathrm{data}[v]
    =
    \bigl\{(p^{(v)}_1,d^{(v)}_1,t^{(v)}_1),\,
            (p^{(v)}_2,d^{(v)}_2,t^{(v)}_2),\,
            \dots\bigr\}.
\]
Each tuple $(p,d,t)$ records that at predicted step $t$ the parent of $v$
became $p$ and its level became $d$. If no tuple satisfies $t\le j$, then the
values from $\hat{T}_0$ apply.

\subsubsection*{Space bound.}
In the incremental-only model, every level change of a vertex is a strict
decrease. Since BFS levels are integers in $\{0,1,\dots,n-1\}$, each vertex can
change parent and level at most $n-1$ times. Therefore, the total number of
stored tuples is at most $n(n-1)$, which gives $O(n^2)$ space.

\subsubsection*{Access cost.}
To recover the state of a single vertex at step $j$, we binary-search
$\mathrm{data}[v]$ for the last tuple with timestamp at most $j$. This costs
$O(\log k_v)$ time, where $k_v=|\mathrm{data}[v]| \le n-1$, hence
$O(\log n)$ in the worst case. Reconstructing one full snapshot naively costs
$O(n\log n)$.

\repeattheorem{spaceOpt}

\section{Space-Efficient  Decremental BFS algorithm}
\label{sec:storage-decremental-fullyDynamic}

To use predictions, we must store the sequence of predicted BFS trees
$\hat{T}_0,\hat{T}_1,\dots,\hat{T}_m$. Naively storing full trees requires
$O(mn)$ space, but we plan to optimize it.

\subsection{Per-Vertex Change Lists}

For each vertex $v\in V$, maintain a chronological list
\[
    \mathrm{data}[v]
    =
    \bigl\{(p^{(v)}_1,d^{(v)}_1,t^{(v)}_1),\,
            (p^{(v)}_2,d^{(v)}_2,t^{(v)}_2),\,
            \dots\bigr\},
\]
where each tuple records that at predicted step $t$ the parent of $v$ became
$p$ and its level became $d$. When no tuple has a timestamp at most $j$, the
state from $\hat{T}_0$ applies.

In the decremental and fully dynamic settings, this list can be interpreted
together with a cursor into $\In{v}$. When a candidate parent becomes invalid,
the cursor advances through $\In{v}$ until a valid parent at level
$\Level{v}-1$ is found. If no such parent exists, then $\Level{v}$ increases
by one and $v$ is placed in the lowering queue.

\subsection{Space Reduction by Skipping Small Pointer Moves}

To save space, we do not store every cursor movement. Fix an integer
$k\ge 1$ and $k\le n$, and store a new tuple for $v$ only when the cursor has advanced by at
least $k$ positions in $\In{v}$ since the previous stored tuple. Equivalently,
we keep only every $k$-th relevant state of the scan.

For a fixed vertex $v$ and a fixed level, this stores at most
$\lceil \dg{v}/k \rceil\leq  \dg{v}/k +1$ tuples. Since a vertex can change level at most
$O(n)$ times, the total storage becomes
\[
    O\!\left(\sum_{v\in V}n (\frac{ \dg{v}}{k}+1)\right)
    = O\!\left(\frac{mn}{k}+n^2\right).
\]
Thus, the space drops by a factor of $k$ compared with storing every
cursor state.

\subsection{Reconstructing a Snapshot}

To recover the tree at step $i$, we load the latest stored state before $i$ and
replay at most $k-1$ skipped cursor moves per affected vertex. Hence, the
reconstruction cost increases by an additive factor of at most
$O(k)$ per vertex whose exact state was not stored, giving worst case cost of $O(\eta_e+\eta^*_v+nk)$.

\repeattheorem{spaceOptDec}


\section{Error Correction}
\label{sec:error-correction}
We now present some trivial and non-trivial error correction mechanisms, reducing the impact of the seemingly stringent definition of error measures. The trivial error correction uses \textit{Set Matching} by checking whether the real updates are a permutation of the predicted updates, where the impact of both is the same, reducing the error to zero. 
For non-trivial error correction, we consider \textit{Prefix Set} for the partially dynamic setting and \textit{Set Edit Distance} for the fully dynamic setting. Recall that the repair cost of our algorithms depends on error measures $\eta_e, \eta_v$ and $\eta^*_v$ relative to
$\hat{T}_i$, where $i$ is the first point of difference. When the real and predicted batches share some common updates beyond $i$, we can start the repair from a later snapshot $\hat{T}_k$ instead of $\hat{T}_i$, where $k>i$. Thus, after an error correction all subsequent updates are processed considering $\hat{T}_k$ instead of $\hat{T}_i$  until the next error correction (if any) is performed. 
Recall that from the divergence point $i$ (first error between $U$ and $\hat{U}$) and current an update step $j > i$,
$U_{i..j} = \{\langle e_{i+1}, t_{i+1}\rangle,\dots, \langle e_j, t_j\rangle\}$ and
$\hat{U}_{i..j} = \{\langle\hat{e}_{i+1}, \hat{t}_{i+1}\rangle,\dots,\langle\hat{e}_j,\hat{t}_j\rangle\}$
denote the real and predicted batch updates.

\subsection{Trivial Error Correction}
If at some point $j>i$, the set of updates $U_{i..j}$ is a permutation of $\hat{U}_{i..j}$, i.e., the set of updates is the same, though their order might differ. Then $\hat{T}_j$ is a valid BFS tree of $G_j$ as both batches contain the same (edge, type) pairs. Hence, we reset the value of $i$ to $j$ for all future updates. 

To evaluate this \textit{Set Matching}  efficiently, we simply check for a current update $<e_j,t_j>$ in $U_{i..j}$ if it has appeared in $\hat{U}_{i..}$, storing the index $k$ of the highest such $<\hat{e}_k,\hat{t}_k>$. If all updates in $U_{i..j}$ have appeared in $\hat{U}_{i...}$ with $k=j$, there exists a one to one match for each updates of ${U}_{i..j}$ in $\hat{U}_{i...}$, so we perform the trivial error correction replacing $i$ with $j$. Clearly, it requires $O(1)$ correction cost while processing each update.

\repeattheorem{trivErr}


\subsection{Non-Trivial Partially Dynamic Error Correction}

In the partially dynamic setting, we match the largest prefix of predicted updates in the current updates to perform \textit{Prefix Set} error correction, accepting the matched predicted updates. Note that in the partially dynamic setting, the presence of even a single incorrectly predicted update renders the corresponding precomputation useless since the result may strongly depend on it in the future.

Let $M=\hat{U}_{i...k}$ be the largest prefix of $\hat{U}_{i...j}$ present in $U_{i...j}$. 
Since the matched updates $M$ were applied to both $G_j$ and $\hat{G}_k$, the two graphs differ only in $U_{i..j}-M$. Hence, in the new batch update, we can simply use $U_{i..j}-M$ instead of the previous $U_{i...j}$.

To perform this error correction efficiently, we maintain a copy of the predicted and real updates in the form of a list. For every real update $<e_j,t_j>$, we locate it in $\hat{U}$ and remove it from both the lists if a match exists. Such an update is added to a temporary list $U_t$ (ordered the same as $\hat{U}$), which is considered along with $U_{i..j}$  while processing an update. Whenever the first element $<\hat{e}_i,\hat{t}_i>$ is removed from $\hat{U}$, the next element $<\hat{e}_{k+1},\hat{t}_{k+1}>$ in $\hat{U}$ can be used for prefix error correction, by simply removing all elements in list $U_t$ up to $<\hat{e}_k,\hat{t}_k>$ (recall list $U_t$ ordered according to $\hat{U}$). The above operations can be performed in $O(1)$ time by storing suitable list pointers as applicable.
Thus, we get the following:

\repeattheorem{nonTrivPDErr}



\begin{remark} 
The corrected edge error $\tilde{\eta}_e=|U_{i..j}|-|M|$ its corresponding vertex errors $\tilde{\eta}_v,\tilde{\eta}^*_v$ (starting from $\hat{T}_k$) satisfies $\tilde{\eta}_e \le \eta_e$, $\tilde{\eta}_v \le \eta_v$ and $\tilde{\eta^*}_v \le \eta^*_v$ always resulting in an improvement, as the change is level of vertices is monotonous in the partially dynamic setting. \end{remark}

\subsubsection{Effect of error correction on space-efficient algorithms.}

After each error correction, an additional cost of $O(n \log n)$ and $O(nk)$ is incurred in the incremental and decremental cases, respectively, for the space-optimized variant. However, we can further improve the total extra cost paid for the incremental case.

\subsubsection*{Efficient search for increasing-time queries.}
\label{sec:method1-search}

When queries arrive in increasing time order, we maintain a pointer into
$\mathrm{data}[v]$ for each $v$ and advance it using exponential
(doubling) jumps ($+1,+2,+4,\dots$) until we reach or overshoot the next query
time. The overshoot is resolved by a binary search over the final interval.
For a monotone query sequence, each stored entry is skipped only a constant
number of times amortized, so the total work over all queries is
$O(\sum_v k_v)=O(n^2)$.
Thus, instead of paying $O(n\log n)$ after each potential $O(m)$ error corrections, we only require total $O(n^2)$ time.

\subsection{Non-Trivial Fully Dynamic Error Correction}

The fully dynamic setting allows us to perform error correction beyond a common prefix set, by choosing an initial target $\hat{T}_k$ minimizing the \textit{set edit distance} between the effective updates. Note that in the fully dynamic setting, any extra updates in $\hat{U}_k$ used to create $\hat{T_k}$ can be removed by applying their complementary updates (insertion instead of deletion and vice versa). Thus, for any $k$ we define matched updates $M_k = U_{i..j} \cap \hat{U}_{i..k}$ resulting in set edit distance $\tilde{\eta}_e=|U_{i..j}|+|\hat{U}_{i..k}|-2|M_k|$. We find the value of $k$ minimizing the set edit distance. Since the matched updates $M_k$ were applied to both $G_j$ and $\hat{G}_k$, the set of updates to be applied to transform $\hat{G}_k$ to $G_j$ are computed as follows:
\begin{align*}
    E^*_{\mathrm{del}} &= \bigl\{(u,v):(u,v,-)\in U_{i..j}-M_k\bigr\} \cup \bigl\{(u,v):(u,v,+)\in \hat{U}_{i..k}-M_k\bigr\}, \\
    E^*_{\mathrm{ins}} &= \bigl\{(u,v):(u,v,+)\in U_{i..j}-M_k\bigr\} \cup \bigl\{(u,v):(u,v,-)\in \hat{U}_{i..k}-M_k\bigr\}.
\end{align*}
Note that minimum set edit distance $\tilde{\eta}_e=|E^*_{\mathrm{del}}|+|E^*_{\mathrm{ins}}|$ is applied to the \textsc{batchDynamicUpdate}$(E^*_{\mathrm{del}},E^*_{\mathrm{ins}})$ from
$\hat{T}_k$. We now describe how to efficiently compute the value of $k$ that minimizes the set edit distance. 

For a given $U_{i..j}$ we need to compute the matched edges $M_k$ for each $\hat{U}_{i..k}$ for $k\in [i,j]$. We do this by incrementally computing $M_k$ from $M_{k-1}$. If $<\hat{e}_k,\hat{t}_k>\in U_{i..j}$, then $M_k=M_{k-1}\cup\{<\hat{e}_k,\hat{t}_k>\}$, else $M_k=M_{k-1}$. This requires $O(1)$ time for each $k$, taking $O(\eta_e)$ time. Now, we compute the minimum set edit distance for $U_{i..j'}$ where $j'=j+1$. If $<e_{j'},t_{j'}>\notin \hat{U}_{i..j'}$ the value of $k$ is unchanged. If $<e_{j'},t_{j'}>=<\hat{e}_{k'},\hat{t}_{k'}>$, again the value of $k$ remains unchanged if $k'<k$, we simply add $<e_{j'},t_{j'}>$ to $M_{k}$. Otherwise, if $k'>k$, the position of the minimum set edit distance can only be in $[k,j']$, which can be incrementally computed in $O(j'-k)$ time as described earlier. 
Note $O(j'-k)=O(\tilde{\eta}_e)$, where $\tilde{\eta}_e$ is the current set edit distance
The corresponding lists $E^*_{del},E^*_{ins}$ can be maintained using suitable pointers. Thus, we can perform error correction in $O(\tilde{\eta}_e)$ time, which is the cost of the current set edit distance. 



\repeattheorem{nonTrivFDErr}

\begin{remark}
The corrected edge error satisfies $\tilde{\eta}_e \le \eta_e$,
so the batch size never increases. However, the corresponding $\tilde{\eta}^*_v$ may be larger
or smaller than $\eta^*_v$ depending on whether the matched updates $M$ moved
predicted levels closer to or further from the true distances in $G_j$. This is unpredictable for the fully dynamic updates.
Non-trivial fully dynamic error correction is therefore most beneficial when $|M|$ is large.
\end{remark}

\begin{remark}
To exploit the smaller set edit distance error metrics, the fully dynamic algorithm can also be used for partially dynamic cases by leveraging the power of the fully dynamic updates as described above. 
\end{remark}




\section{Conclusion}
\label{sec:conclusion}

We have presented algorithms using the prediction models for maintaining a BFS tree
in a directed dynamic graph, covering the incremental, decremental, and fully
dynamic settings. Our algorithms exploit a predicted update sequence, which is
preprocessed offline to achieve running times that scale with prediction
quality rather than worst-case graph parameters. Specifically, the update
time is $O(1)$ when predictions are perfect, $O(\eta_e + \eta_v)$ for
incremental updates, and $O(\min(\eta_e + \eta^*_v,\, m))$ for decremental and
fully dynamic updates — never exceeding the classical $O(m)$ bound regardless
of prediction quality. Moreover, the total update time for both incremental and decremental algorithms is also bounded by $O(\eta_e + \eta^*_v)$ - never exceeding the classical $O(mn)$ total time of ES Trees~\cite{even-shiloach}. Thus, our algorithms satisfy the desired properties of consistency, robustness, and smooth degradation of the algorithms with the prediction model.
We also presented a space-efficient storage schemes for the incremental algorithm
that reduce the naive $O(mn)$ space to $O(n^2)$ at the expense of additional time complexity of $O(n\log n)$ in the worst case and total update time. Similarly, we presented a space-efficient storage scheme for the decremental algorithm
that reduces the $O(mn)$ space to $O(mn/k+n^2)$ at the expense of additional time complexity of $O(nk)$ in the worst case and total update time.
Finally, we presented some trivial and non-trivial error correcting mechanisms to reduce the dependence of our error measures on the first difference between predicted and real updates, to Set Matching, Prefix Set, and Set Edit Distance based error measures. These improved error measures directly impact the theoretical bounds and expected practical performance of our algorithms. 

Several directions remain open for future work. Firstly, to implement a full benchmark suite to evaluate our algorithms on
real dynamic graph datasets and measure the percentage improvement in update
time over classical approaches, particularly with respect to prediction error measures. 
On the theoretical side, it would be interesting to
study \emph{adaptive} prediction models that are refined online as real updates
arrive, potentially tightening the error measures over time. Finally, our algorithms at their core are simple extensions of the classical ES Trees for handling batch updates, which proves sufficient to give a prediction algorithm satisfying \textit{robustness, consistency, and smooth degradation}. Thus, our framework can be generalized to other dynamic graph problems by using their corresponding batch-update versions.



\bibliography{ESA}

\newpage

\appendix



    








\end{document}